\documentclass[11pt,twoside]{article}
 \usepackage[table]{xcolor}
\usepackage{array}
\usepackage{tabulary}
\usepackage{a4wide}
\definecolor{light-gray}{gray}{0.965}
\usepackage{amsfonts,amsmath,amssymb,amsthm, placeins}
\usepackage{color,graphics}
\usepackage{courier}
\usepackage{calligra}
\DeclareMathAlphabet{\mathcalligra}{T1}{calligra}{m}{n}
\newcolumntype{C}[1]{>{\centering\let\newline\\\arraybackslash\hspace{0pt}}m{#1}}
 \usepackage{makecell}
 \usepackage{tabularx}
 \newcolumntype{b}{X}
\newcolumntype{s}{>{\hsize=.5\hsize}X}
 \usepackage{bbm}
 \usepackage{array}
\newcolumntype{C}{>{\centering\arraybackslash}X}
\usepackage[latin1]{inputenc}
\usepackage[dvips]{graphicx}
\usepackage{booktabs}
\usepackage{caption}
\usepackage[english]{babel}
\usepackage{xcolor}

\usepackage{subfig}
\definecolor{darkgreen}{rgb}{0.1,0.1,1}
\usepackage{bigstrut}
\usepackage{multirow}
\usepackage{longtable}
\usepackage{array}
\usepackage{url}
\usepackage{rotating}
\newsubfloat{figure}
\usepackage{float}
\usepackage{amsmath}
\numberwithin{equation}{section}

\usepackage[round,comma]{natbib}

\normalsize
\newcommand*{\lcterm}[3]{%
\settoheight{\LC@temph}{\mbox{$\scriptstyle\overline{#3\,}$}}
\settowidth{\LC@tempw}{$\scriptstyle #2$}
{#1}{}_{#2:\lcroof{#3}}^{\makebox[\LC@tempw]{\hfill$\scriptstyle 1$\hfill}}
}

\makeatletter
\newcommand*{\rom}[1]{\expandafter\@slowromancap\romannumeral #1@}
\makeatother

\newtheorem{theorem}{Theorem}

\newtheorem{corollary}{Corollary}

\newtheorem{example}{Example}

\newtheorem{lemma}[theorem]{Lemma}

\newtheorem{proposition}{Proposition}
\newtheorem{remark}{Remark}

\usepackage{amssymb}
\usepackage{dsfont}
\usepackage{amsmath}
\numberwithin{theorem}{section}
\numberwithin{remark}{section}
\numberwithin{corollary}{section}
\numberwithin{proposition}{section}
\numberwithin{example}{section}

\usepackage[pdftex,bookmarks=true,colorlinks=true,urlcolor=blue,linkcolor=blue,citecolor=blue,breaklinks=true,a4paper]{hyperref}
\usepackage[capitalise]{cleveref}

\Crefformat{table}{#2{}{Table} #1{}#3}

\date{}
\usepackage{authblk}

\makeatletter
\renewcommand\AB@affilsepx{;\protect\Affilfont}
\makeatother

\usepackage{titlesec}
\usepackage{sectsty}

\title{Risk aggregation and capital allocation using a new generalized Archimedean copula}
\author[a]{Fouad Marri \thanks{Corresponding Author: fmarri@insea.ac.ma}}
\author[b]{Khouzeima Moutanabbir}
 \affil[a]{Department of Statistics and Actuarial Science, Institut National de Statistique et d'Economie Appliqu\'ee, INSEA, Morocco}\affil[b]{Department of Finance and Investment Management, School of Management, University of Johannesburg, Johannesburg, South Africa.} 

\begin{document}
\bibliographystyle{myapalikecopie}
\maketitle
\begin{abstract}
In this paper, we address risk aggregation and capital allocation problems in the presence of dependence between risks. The dependence structure is defined by a mixed Bernstein copula which represents a generalization of the well-known Archimedean copulas. Using this new copula, the probability density function and the cumulative distribution function of the aggregate risk are obtained. Then, closed-form expressions for basic risk measures, such as tail value-at-risk (TVaR) and TVaR-based allocations, are derived. \newline

$\emph{Keywords}$ : \emph{Bernstein copulas}; \emph{Capital allocation}; \emph{Copulas}; \emph{Dependence}; \emph{Tail value at risk}; \emph{Value-at-Risk} 
\end{abstract}
\section{Introduction\label{sectA}}
Risk aggregation and risk-based capital allocation have attracted considerable attention in actuarial sciences and quantitative risk management over the past years. One of the crucial applications of risk aggregation is to determine the required regulatory capitals and to price insurance and reinsurance products. For this purpose, an adequate risk measure should be used to evaluate the whole level of risk for a given portfolio. An important risk measure is the Value at Risk (VaR), which is defined as a threshold value such that the probability the loss on a portfolio exceeds this value is a given probability $1-\kappa$. Another interesting risk measure is the Tail Value at Risk (TVaR), also known as the conditional tail expectation (CTE), which represents the average amount of a loss given that the loss exceeds a specified quantile. The TVaR is known to be a coherent risk measure over the space of continuous random variables (see \citet{dhaene2008can}; \citet{furman2008weighted}; \citet{furman_landsman_2006}).

In this paper, we address risk aggregation and capital allocation using the TVaR risk measure. Consider a vector of $n$ continuous and non-negative random variables (rv's) $\mathbf{X}=\left(X_{1},\cdots,X_{n}\right)$, the component $X_i$ denotes the marginal risk (claim or loss). We define the aggregate loss as $S_n=X_1+\cdots+X_n$. For a given confidence level $\kappa$, the value at risk of $S_n$, $VaR_\kappa(S_n)$, is defined by
\begin{equation*}
VaR_\kappa(S_n)=inf \left (x \in \mathbb{R}, F_{S_n}(x)\geq \kappa\right),
\end{equation*}
where $0 <\kappa< 1$ and $F_{S_n}$ is the cumulative distribution function (cdf) of $S_{n}$. We also define the tail value at risk of $S_n$, $TVaR_\kappa(S_n)$, as follows
\begin{equation}
TVaR_\kappa(S_n)=E\left (S_n\mid S_n> VaR_\kappa(S_n)\right). \label{fuaredsnb}
\end{equation}
The contribution of the $i$-th risk $X_i$ to the aggregate risk $S_n$ is given by
\begin{eqnarray}
TVaR_\kappa(X_i;S_n)&=&E\left (X_i \mid S_n> VaR_\kappa(S_n)\right), \label{fuaredsnbr}
\end{eqnarray}
for $i=1,\cdots,n.$ The additivity of the expectation allows the decomposition of the TVaR into the sum of TVaR contributions as follows
\begin{equation*}
TVaR_\kappa(S_n)=\sum\limits_{i=1}^{n} TVaR_\kappa(X_i;S_n).
\end{equation*}
Adequate economic capital and capital allocation relay on accurate dependences modeling between different components of the portfolio. For this reason, multivariate risk models incorporating dependence are very important in risk modeling in finance and insurance. In the literature, several classes of multivariate distributions have been proposed. \citet{landsman2003tail} obtained the explicit formulas of $TVaR_\kappa(S_n)$ and $TVaR_\kappa(X_i;S_n)$ for the multivariate elliptical distributions, which include the distributions such as multivariate normal, stable, student, etc. Other closed-form expressions for the economic capital and risk contribution under multivariate phase-type
distributed risks have been given in \citet{cai2005conditional}. The case of multivariate gamma
distribution for risks has been studied in \citet{furman2005risk} as well as a multivariate Tweedie distribution in \citet{furman2010multivariate}. \citet{chirag2007} consider the case of multivariate Pareto risks. In these papers, explicit expressions for the $TVaR_\kappa(S_n)$ and the TVaR-based allocation are derived. For further details on the TVaR-based
allocation of risk capital, \citet{kim2007estimation} and references within can be consulted. Other researchers addressed the risk aggregation problem using copulas. For example, \citet{cossette2013multivariate} consider a portfolio of dependent risks whose multivariate distribution is the Farlie-Gumbel-Morgenstern copula with mixed Erlang marginal distributions. \citet{sarabia2016risk} give explicit formulas for the probability density function of $S_n$ for some multivariate mixed exponential distributions for which the dependence structure is an Archimedean copula. Recently, \citet{marri2018moments} derive explicit expressions for the higher moments of the discounted aggregate renewal claims with dependence. 

The focus of this paper is to derive explicit formulas of $TVaR_\kappa(S_n)$ and $TVaR_\kappa(X_i;S_n)$ for $i =1,\cdots,n$, with a more general dependence structure that allows us to capture different types of dependence between structure. The multivariate model that we suggest is a mixed Bernstein copula. The mixed Bernstein copulas have many attractive properties and, in particular, they are non-exchangeable. This is useful for risk aggregation in many insurance and financial applications. Our model provides a new generalization of the well-known Archimedean copulas. The remainder of this paper is structured as follows. In Section \eqref{sectB}, we introduce the mixed Bernstein copulas. The distribution of the aggregated risk is investigated in Section \eqref{sectC}. In Section \eqref{sectDD}, we derive closed-formulas of the $TVaR_\kappa(S_n)$ and $TVaR_\kappa(X_i;S_n)$ for $n$ dependent rv's $X_1, X_2,\cdots,X_n$ joined by the mixed Bernstein copulas. We obtain specific expressions for the
aggregated distribution in Section \eqref{sectD}. The results are illustrated with numerical applications in Section \eqref{sectEE}.
\section{Mixed Bernstein copulas\label{sectB}}
In this section, we will construct a new family of copulas that will be used in the different models.
 
Let $\mathbf{X}=\left(X_{1},\cdots,X_{n}\right)$ be a vector of $n$ continuous and positive random variables (rv's) with joint survival distribution function (sf) denoted by $\bar{H}$ and univariate survival marginal distributions $\bar{H}_i,\quad i=1,\cdots,n$. Let $\Theta$ be a positive random variable (rv) with probability density function (pdf) $f_{\Theta}$, cumulative distribution function (cdf) $F_{\Theta}$, and Laplace transform $f_{\Theta}^{\star}(s)=\int_{0}^{\infty} e^{-s\theta }f_{\Theta} (\theta)d\theta$. In this paper, we assume that $X_{1},\cdots,X_{n}$ are $n$ dependent, positive and continuous rv's such that 
 \begin{eqnarray}
{ \left(X_{1},\cdots,X_{n}\right)^{\top}=\left(\frac{Z_{1}}{\Theta},\cdots,\frac{Z_{n}}{\Theta}\right)^{\top}}, \label{eqP2derfmlo}
\end{eqnarray}
where $\left(Z_{1},\cdots,Z_{n}\right)$ is a vector of $n$ continuous rv's with joint survival distribution function denoted by $\bar{F}$ and with standard exponential marginal distributions (with mean $1$).

 In \citet{sarabia2018aggregation} and \citet{Albrecher11}, the rv's $Z_i$ are supposed to be independent. While such an assumption significantly simplifies the model setup, it is also to possible that it leads to a misidentification of the dependence structure. Indeed, the variable $\Theta$ only capture the common factor of dependence between all the n variables (e.g., climate conditions, age,$\cdots$,etc.). In this paper, we add another level of dependence by assuming the vector $Z_i$ has dependent components. According to Sklar's theorem for survival functions (see e.g. \citet{Sklar}), the multivariate survival function of ${ Z_{1}},\cdots,{ Z_{n}}$ can be written as a function of the marginal survival functions and a copula $C_1$ describing the dependence structure as follows:
$\bar{F}(z_1,\cdots,z_n)={C_1}\left(e^{-z_1},\cdots,e^{-z_n}\right),$ for $n \in \left\{2,3,\cdots \right\},\,$ and $z_1,\cdots,z_n \geq0.$

In this paper, we assume that $C_1$ is defined with the Bernstein copulas $C_{B} $ introduced in \citet{sancetta2004bernstein} and defined under some specifics conditions on the function $\alpha$, by
 \begin{eqnarray}
C_{B}(u_1,\cdots,u_n)
&=&\displaystyle\sum\limits_{\nu_1=0}^{m_1}\cdots\displaystyle\sum\limits_{\nu_n=0}^{m_n}\alpha\left(\frac{\nu_1}{m_1},\cdots,\frac{\nu_n}{m_n}\right)\prod\limits_{i=1}^{n}G_{\nu_i:m_i}(u_i),\label{eqPremiumrate1}
\end{eqnarray}
for every $(u_1,\cdots,u_n)$ in $\left[ 0,1\right]^{n}$ such that $ 0\leq \nu_i\leq m_i\in \mathbb{N}$, where
 \begin{eqnarray}
G_{\nu_i:m_i}(u_i)&=&\binom{m_i}{\nu_i} u_i^{\nu_i}(1-u_i)^{m_i-\nu_i},\quad \nu_i=0,1,\cdots,m_i, \label{eqPremiumrate1er}
\end{eqnarray}
is the $\nu_i$th Bernstein polynomial of order $m_i,\quad i=1,\cdots,n$. \citet{sancetta2004bernstein} showed that the coefficients of the Bernstein copulas $C_{B}$ have a direct interpretation as
the points of some arbitrary approximated copula $C_1$, i.e., $C_1\left(\frac{\nu_1}{m_1},\cdots,\frac{\nu_n}{m_n}\right)=\alpha\left(\frac{\nu_1}{m_1},\cdots,\frac{\nu_n}{m_n}\right).$ This justifies our choice of Bernstein copulas since our mixed Bernstein model could approximate many different dependence structures. In \citet{cottin2014bernstein}, it has been proved that any copula function can be approximated uniformly using Bernstein polynomials. In the following theorem, we give sufficient conditions for $C_{B}(u_1,\cdots,u_n)$ to be a copula.
\begin{theorem}\label{ritazcmlopdes} 
 The Bernstein polynomial $C_{B}(u_1,\cdots,u_n)$ is a copula function if conditions
 \begin{eqnarray}
\displaystyle\sum\limits_{l_1=0}^{1}\cdots \displaystyle\sum\limits_{l_n=0}^{1}(-1)^{n+l_1+\cdots+l_n} \alpha\left(\frac{\nu_1+l_1}{m_1},\cdots,\frac{\nu_n+l_n}{m_n}\right)&\geq &0,
\end{eqnarray}
for all $ 0\leq \nu_i\leq m_i-1,\quad i=1,\cdots,n,$ \label{ritazcmlop346sq9f} 
 \begin{eqnarray}
\alpha\left(\frac{\nu_1}{m_1},\cdots,\frac{\nu_{i-1}}{m_{i-1}},0,\frac{\nu_{i+1}}{m_{i+1}}\cdots,\frac{\nu_n}{m_n}\right)=0,\quad \forall i=1,\cdots,n, \label{ritazcmlop3469sz} 
\end{eqnarray}
and
 \begin{eqnarray}
\alpha\left(1,\cdots,1,\frac{\nu_{i}}{m_{i}},1,\cdots,1\right)=\frac{\nu_{i}}{m_{i}},\quad \forall i=1,\cdots,n \label{ritazcmlop346sq9fgr} 
\end{eqnarray}
hold. Moreover, \eqref{ritazcmlop3469sz} and \eqref{ritazcmlop346sq9fgr} are necessary for $C_{B}(u_1,\cdots,u_n)$ to be a copula.
\end{theorem}
\begin{proof}A proof of this theorem was given by \citet{yang2020bernstein}.\end{proof} 
 The application of Bernstein copulas in actuarial science is recent. \citet{salmon2006pricing} and \citet{hurd2007using} apply Bernstein copulas to the pricing of two-asset derivatives written on foreign exchange rates. \citet{diers2012dependence} use Bernstein copulas to model the dependence between non-life insurance risks. \citet{tavin2013application} analyze properties of Bernstein copulas in a context of multi-asset derivatives pricing. In contrast to Archimedean copulas, Bernstein copulas can model non-exchangeable dependence structures.
 
 From \citet{sancetta2004bernstein} and \citet{cottin2014bernstein}, the corresponding pdf of the copula $C_{B}$ is given by
\begin{eqnarray}
c_{B}(u_1,\cdots,u_n)
&=&\displaystyle\sum\limits_{\nu_1=0}^{m_1-1}\cdots \displaystyle\sum\limits_{\nu_n=0}^{m_n-1}
\gamma\left(\frac{\nu_1}{m_1},\cdots,\frac{\nu_n}{m_n}\right)\left(\prod\limits_{i=1}^{n}m_iG_{\nu_i:m_i-1}(u_i)\right),\label{eqPremiumrate2}
\end{eqnarray}
where
\begin{eqnarray}
\gamma\left(\frac{\nu_1}{m_1},\cdots,\frac{\nu_n}{m_n}\right)&=&\displaystyle\sum\limits_{l_1=0}^{1}\cdots \displaystyle\sum\limits_{l_n=0}^{1}(-1)^{n+l_1+\cdots+l_n} \alpha\left(\frac{\nu_1+l_1}{m_1},\cdots,\frac{\nu_n+l_n}{m_n}\right).\nonumber
\end{eqnarray}
\citet{cottin2014bernstein} show that the Bernstein copulas density function can also be expressed as
\begin{eqnarray}
c_{B}(u_1,\cdots,u_n)&=&\displaystyle\sum\limits_{\nu_1=0}^{m_1-1}\cdots \displaystyle\sum\limits_{\nu_n=0}^{m_n-1}\Pr\left(N_1=\nu_1,\cdots,N_n=\nu_n\right)\left(\prod\limits_{i=1}^{n}m_iG_{\nu_i:m_i-1}(u_i)\right),\label{eqPremideqmfre}
\end{eqnarray}
where $\left(N_1,\cdots,N_n\right)$ is a random vector whose marginal component $N_i$ follows a discrete uniform distribution over $ \left\{0,1,\cdots,m-1 \right\}$ and $\Pr\left(N_1=\nu_1,\cdots,N_n=\nu_n\right)=\gamma\left(\frac{\nu_1}{m_1},\cdots,\frac{\nu_n}{m_n}\right).$ For more details about Bernstein copulas, we refer
readers to \citet{kulpa1999approximation} and \citet{sancetta2004bernstein}.

From \eqref{eqP2derfmlo}, the joint survival function of $X_{1},\cdots,X_{n}$ can be written as
\begin{eqnarray}
\bar{H}\left(x_1,\cdots,x_n\right)&=&\int_{0}^{\infty}C_{B}\left(e^{-\theta x_1},\cdots,e^{-\theta x_n} \right) f_{\Theta}(\theta)d\theta, \label{eqP}
\end{eqnarray}
for $x_1,\cdots,x_n \geq0$. It implies that the marginal survival function of ${ X_{i}}$ is given by 
\begin{eqnarray}
\bar{H}_i(x)=\Pr({X_i}\geq x )&=&f_{\Theta}^{\star} (x),\label{fuarezdef}
\end{eqnarray}
for $i=1,\cdots,n,$ where $f_{\Theta }^{\star}$ is the Laplace transforms of $F_{\Theta }$. Note that  the marginal random variables $X_i$ are necessarily completely monotone (see, e.g.,\citet{oakes1989bivariate}).

A closed-form expression for the survival function of $(X_{1},\cdots,X_{n})$ is given in the next theorem.
\begin{theorem}\label{ritazcmlop}Let $\mathbf{X}=\left(X_{1},\cdots,X_{n}\right)$ be a vector of $n$ continuous and non-negative rv's defined by the stochastic representation \eqref{eqP2derfmlo}. Then the survival function of $(X_{1},\cdots,X_{n})$ is given by
\begin{eqnarray}
\bar{H}\left(x_1,\cdots,x_n\right)&=&\displaystyle\sum\limits_{\ell_1=0}^{m_1}\cdots\displaystyle\sum\limits_{\ell_n=0}^{m_n} \beta_{\ell_1,\cdots,\ell_n} f_{\Theta}^{\star}\left(\sum\limits_{i=1}^{n} \ell_i x_i\right), \label{ritazcmlop3469} 
\end{eqnarray}
for $x_1,\cdots,x_n \geq0$, where
$\beta_{\ell_1,\cdots,\ell_n}=\displaystyle\sum\limits_{\nu_1=0}^{\ell_1}\cdots\displaystyle\sum\limits_{\nu_n=0}^{\ell_n}(-1)^{\sum\limits_{i=1}^{n}(\ell_i-\nu_i)}\left[\prod\limits_{i=1}^{n}\binom{m_i-\nu_i}{m_i-\ell_i} \right]\left[\prod\limits_{i=1}^{n}\binom{m_i}{\nu_i} \right]\alpha\left(\frac{\nu_1}{m_1},\cdots,\frac{\nu_n}{m_n}\right)
.$
\end{theorem}
\begin{proof}From \eqref{eqPremiumrate1} and \eqref{eqP}, we have 
\begin{eqnarray}
\bar{H}\left(x_1,\cdots,x_n\right)&=&\displaystyle\sum\limits_{\nu_1=0}^{m_1}\cdots\displaystyle\sum\limits_{\nu_n=0}^{m_n}\alpha\left(\frac{\nu_1}{m_1},\cdots,\frac{\nu_n}{m_n}\right)\int_{0}^{\infty}\prod\limits_{i=1}^{n}G_{\nu_i:m_i}(e^{-\theta x_i}) f_{\Theta}(\theta)d\theta.\label{rezromiofdmp}
\end{eqnarray}
Otherwise, applying the binomial theorem to $(1-e^{-\theta x_i})^{m_i-\nu_i}$ and using \eqref{eqPremiumrate1er} yield
 \begin{eqnarray*}
 \prod\limits_{i=1}^{n}G_{\nu_i:m_i}\left(e^{-\theta x_i}\right)&=&\left[\prod\limits_{i=1}^{n}\binom{m_i}{\nu_i} \right]\displaystyle\sum\limits_{\ell_1=0}^{m_1-\nu_1}\cdots\displaystyle\sum\limits_{\ell_n=0}^{m_n-\nu_n}\left[\prod\limits_{i=1}^{n}\binom{m_i-\nu_i}{\ell_i} \right] e^{-\theta \sum\limits_{i=1}^{n} x_i\left({\nu_i}+\ell_i\right)}\left(-1\right)^{\displaystyle\sum\limits_{i=1}^{n}\ell_i}.
\end{eqnarray*}
Substituting the last expression into \eqref{rezromiofdmp}, we obtain
\begin{eqnarray*}
&&\bar{H}\left(x_1,\cdots,x_n\right)\nonumber\\
&=&\displaystyle\sum\limits_{\nu_1=0}^{m_1}\cdots\displaystyle\sum\limits_{\nu_n=0}^{m_n}\alpha\left(\frac{\nu_1}{m_1},\cdots,\frac{\nu_n}{m_n}\right)\left[\prod\limits_{i=1}^{n}\binom{m_i}{\nu_i} \right] \displaystyle\sum\limits_{\ell_1=0}^{m_1-\nu_1}\cdots\displaystyle\sum\limits_{\ell_n=0}^{m_n-\nu_n}\left[\prod\limits_{i=1}^{n}\binom{m_i-\nu_i}{\ell_i} \right](-1)^{\sum\limits_{i=1}^{n}\ell_i} f_{\Theta}^{\star}\left(\sum\limits_{i=1}^{n} \left(\nu_i+\ell_i\right)x_i\right)\nonumber\\
&=&\displaystyle\sum\limits_{\nu_1=0}^{m_1}\cdots\displaystyle\sum\limits_{\nu_n=0}^{m_n}\alpha\left(\frac{\nu_1}{m_1},\cdots,\frac{\nu_n}{m_n}\right)\left[\prod\limits_{i=1}^{n}\binom{m_i}{\nu_i} \right] \displaystyle\sum\limits_{\ell_1=\nu_1}^{m_1}\cdots\displaystyle\sum\limits_{\ell_n=\nu_n}^{m_n}\left[\prod\limits_{i=1}^{n}\binom{m_i-\nu_i}{\ell_i-\nu_i} \right](-1)^{\sum\limits_{i=1}^{n}(\ell_i-\nu_i)} f_{\Theta}^{\star}\left(\sum\limits_{i=1}^{n} \ell_i x_i\right).\label{eqPfdfrfg}
\end{eqnarray*}
Inverting the order of summation with respect to the variables $\ell_i$ and $\nu_i,$ $i=1,\cdots,n$, one readily obtains 
\begin{eqnarray*}
\bar{H}\left(x_1,\cdots,x_n\right)&=&\displaystyle\sum\limits_{\ell_1=0}^{m_1}\cdots\displaystyle\sum\limits_{\ell_n=0}^{m_n} \beta_{\ell_1,\cdots,\ell_n} f_{\Theta}^{\star}\left(\sum\limits_{i=1}^{n} \ell_i x_i\right),
\end{eqnarray*}
from which we get the desired result.
\end{proof} 
 On the other hand, according to Sklar's theorem for survival functions, see e.g. \citet{Sklar}, the joint survival function of $X_{1},\cdots, X_{n}$ can be written as a function of the marginal survival functions $\bar{H}_i,\, i=1,\cdots,n,$ and a copula $C$ describing the dependence structure as follows:
\begin{eqnarray}
\bar{H}\left(x_1,\cdots,x_n\right)&=&{C}\left(\bar{H}_1(x_1),\cdots,\bar{H}_n(x_n)\right),\label{treyfre}
\end{eqnarray}
for $n \in \left\{2,3,\cdots \right\},\,$ and $x_1,\cdots,x_n \geq0.$ Then, the following proposition holds.
\begin{proposition}\label{ritazcmlop34} The copula function $C : [0, 1]^n \to [0, 1]$ that corresponds to the general dependence structure defined in Theorem \eqref{ritazcmlop} is given, for $u_i \in [0, 1], i = 1,\cdots,n,$ by
\begin{eqnarray}
C\left( u_1,\cdots,u_n \right)&=&\displaystyle\sum\limits_{\ell_1=0}^{m_1}\cdots\displaystyle\sum\limits_{\ell_n=0}^{m_n}\beta_{\ell_1,\cdots,\ell_n} f_{\Theta}^{\star}\left(\sum\limits_{i=1}^{n} \ell_i f_{\Theta}^{\star -1}(u_i) \right), \label{ritazcmlop346h} 
\end{eqnarray}
where
$\beta_{\ell_1,\cdots,\ell_n}=\displaystyle\sum\limits_{\nu_1=0}^{\ell_1}\cdots\displaystyle\sum\limits_{\nu_n=0}^{\ell_n}(-1)^{\sum\limits_{i=1}^{n}(\ell_i-\nu_i)}\left[\prod\limits_{i=1}^{n}\binom{m_i-\nu_i}{m_i-\ell_i} \right]\left[\prod\limits_{i=1}^{n}\binom{m_i}{\nu_i} \right]\alpha\left(\frac{\nu_1}{m_1},\cdots,\frac{\nu_n}{m_n}\right).$
\end{proposition}
\begin{proof}The result follows easily from \eqref{fuarezdef}, \eqref{ritazcmlop3469} and \eqref{treyfre}. 
\end{proof}
This new family of copulas $C$ extends the well-known Archimedean copulas and could be seen as a mixture of distorted Archimedean copulas with generator $f_{\Theta}$. The study of the properties of this copula are beyond the scoop of this paper and we leave it for a future research.
\begin{remark} From \eqref{ritazcmlop3469sz}, one gets $\beta_{\ell_1,\cdots,\ell_{i-1},0,\ell_{i+1},\cdots,\ell_n}=0,$ for $i=1,\cdots,n.$
\end{remark}
\begin{remark} If $m_1=\cdots=m_n=1$, then the copula $C$ in \eqref{ritazcmlop346h} reduces to a n-dimensional Archimedean copula with generator $f_{\Theta}$ and given by $C\left( u_1,\cdots,u_n \right)= f_{\Theta}^{\star}\left(\sum\limits_{i=1}^{n} f_{\Theta}^{\star -1}(u_i) \right)$.
\end{remark}
\begin{remark} The most simple form of the bivariate Bernstein copula $C_B$ is given by the Farlie-Gumbel-Morgenstern copula, which is defined by
$C_{B}^{FGM}(u,v)=uv+\delta uv(1-u)(1-v),\,$ for $\delta\in[-1,1]$ with $\delta=4\alpha(\frac{1}{2},\frac{1}{2})-1$ and $m_1=m_2=n=2$. The expression for $C\left( u_1,u_2 \right) $ in \eqref{ritazcmlop346h} turns into
\begin{eqnarray}
C\left( u_1,u_2 \right)&=&\displaystyle\sum\limits_{\ell_1=1}^{2}\displaystyle\sum\limits_{\ell_2=1}^{2}\beta_{\ell_1,\ell_2} f_{\Theta}^{\star}\left(\sum\limits_{i=1}^{2} \ell_i f_{\Theta}^{\star -1}(u_i) \right), \label{ritazcmlop34fr6h} 
\end{eqnarray}
where
$\beta_{1,1}=1+\delta,\,$ $\beta_{1,2}=\beta_{2,1}=-\delta\,$ and $\beta_{2,2}=\delta.$ Then $C$ in \eqref{ritazcmlop34fr6h} reduces to the copula discussed in \citet{cote2019dependence}.
\end{remark}
For the sake of simplicity, it is assumed that $m_i = m,$ for $i = 1,\cdots, n$.

An appropriate choice of the joint cumulative distribution $\alpha$ in the calculation of the coefficient $\beta$ could lead to a different dependence structure. This point is illustrated in the following example. In fact, the copulas defined in Proposition \eqref{ritazcmlop34} could significantly change the obtained dependence by an Archimedean copula.
\begin{example}
 Assume that the mixing rv $\Theta$ is following a Gamma distribution such that
 \begin{equation*}
 f_{\Theta}^{\star} (x)=\left(1+\frac{x}{b}\right)^{-a},\quad a \geq 1,
 \end{equation*}
 Thus, when $m=1$ the copula $C$ in \eqref{ritazcmlop346h} is reduced to a Clayton copula. For different values of $m$ and for different choices of $\alpha$, the dependence structure is assessed via the Spearman's $\rho$ that is given by
 \begin{equation}\label{spearman}
 \rho=12\int_{0}^{1} \int_{0}^{1} C(u_{1},u_{2})du_{1} du_{2}-3.
 \end{equation}
 This allows us to measure the impact of introducing the Bernstein copula on the dependence structure of the mixed exponential model. In this illustrations, we consider two cases for $alpha$
 \begin{description}
 \item[(i)] Counter-comonotonic: $\alpha(u_{1},u_{2})=max(u_{1}+u_{2}-1,0)$, and
 \item[(ii)] Comonotonic: $\alpha(u_{1},u_{2})=min(u_{1},u_{2})$.
 \end{description}
 Given this choices for $\alpha$, the obtained values of $\rho$ will consist of an upper bound if $\alpha$ is comonotonic and a lower limit if $\alpha$ is counter-comonotonic. In Figure \ref{Fig1}, the values of $\rho$ are displayer for $m=1,2,\cdots,15$ and for $a=1,5$, and 10.
 \begin{figure}[h]
 \centering
 \includegraphics[width=1\linewidth]{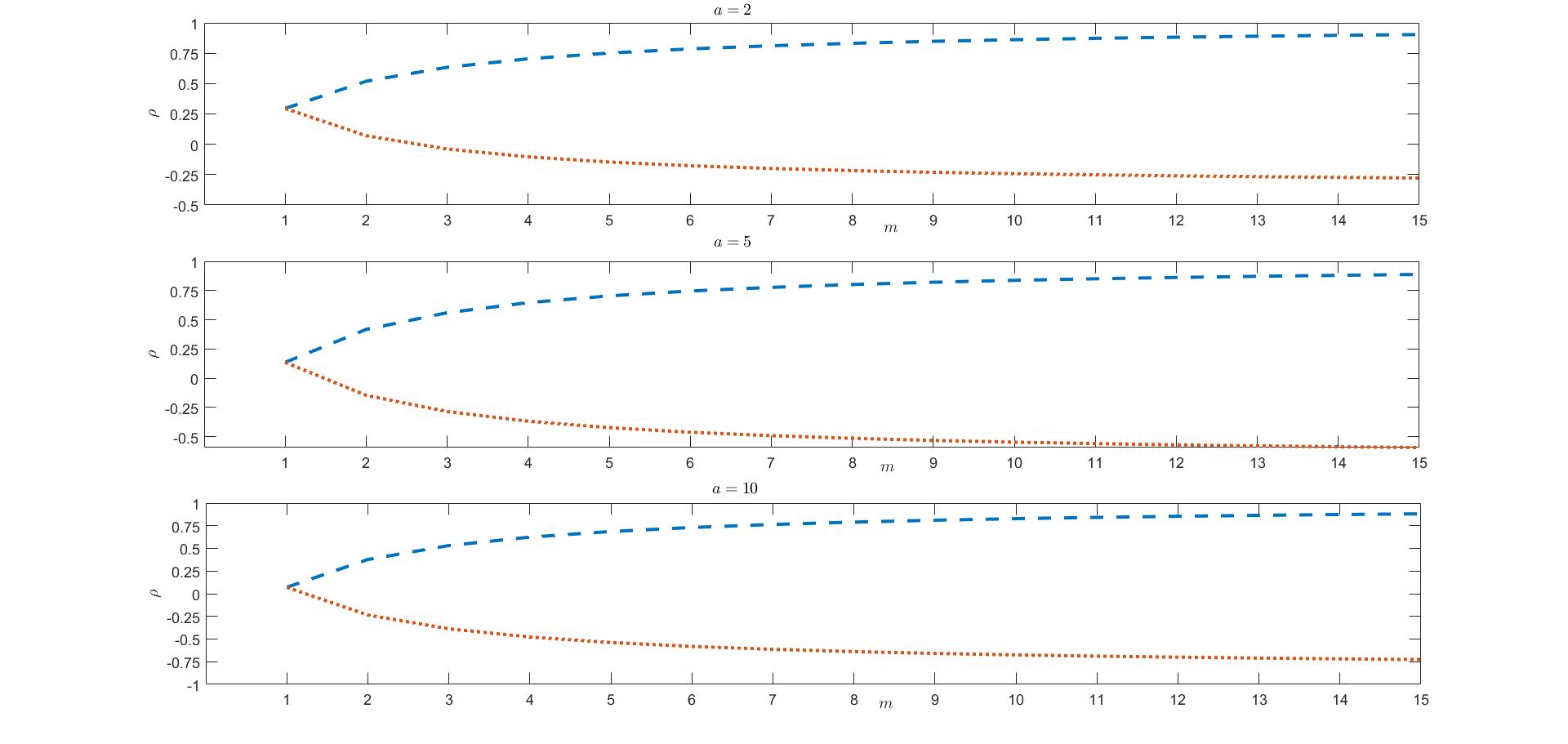}
	 \caption[]{\label{Fig1} The upper and lower bounds for $\rho$ for different values for $m$. }
\end{figure}
It is obvious that the $\rho$ decreases (increases) as $m$ increases when the Bernstein copula has a negative (positive) dependence. Moreover, the dependence is significantly changing, which means that the introduction of the mixed-Bernstein copula improves the obtained range of dependence and allows $\rho$ to reach values beyond the value in the case of Archimedean copula (when $m=1$). Similar results with the same pattern are obtained under different choices of the copula $\alpha$.
\end{example}
\section{The distribution of the aggregated risk ${S}_n=\sum\limits_{i=1}^{n}X_i$\label{sectC}}
In this section, we obtain the probability density function and the survival
function of the aggregated risk $S_n$.

The joint probability density function (pdf) of $\left({ Z_{1}},\cdots,{ Z_{n}}\right)$ is given by 
 \begin{eqnarray}
{f}_{{ Z_{1}},\cdots,{ Z_{n}}}\left(z_1,\cdots,z_n\right)&=&{c_B}\left(e^{-z_1},\cdots,e^{-z_n}\right)e^{-\sum\limits_{i=1}^{n} z_i}.\label{treyez}
\end{eqnarray}
Combining  \eqref{eqPremiumrate2} and \eqref{treyez}, the pdf of $\left({ Z_{1}},\cdots,{ Z_{n}}\right)$ becomes
\begin{eqnarray}
f_{Z_{1},Z_{2},\cdots,Z_{n}}(z_1,\cdots,z_n)
&=&m^n e^{-\sum\limits_{i=1}^{n} z_i}\displaystyle\sum\limits_{\nu_1=0}^{m-1}\cdots \displaystyle\sum\limits_{\nu_n=0}^{m-1}
\gamma\left(\frac{\nu_1}{m},\cdots,\frac{\nu_n}{m}\right)\prod\limits_{i=1}^{n}\left[\binom{m-1}{\nu_i} e^{-\nu_i z_i}(1-e^{-z_i})^{m-1-\nu_i}\right]\nonumber\\
&=&\displaystyle\sum\limits_{\nu_1=0}^{m-1}\cdots \displaystyle\sum\limits_{\nu_n=0}^{m-1}
\gamma\left(\frac{\nu_1}{m},\cdots,\frac{\nu_n}{m}\right)\prod\limits_{i=1}^{n} f_{W_i}(z_i), \label{23s}
\end{eqnarray}
where $W_i$ are $n$ independent rv's with pdf $f_{W_i}(z)=m\binom{m-1}{\nu_i} e^{-(\nu_i+1) z}(1-e^{-z})^{m-1-\nu_i}.$ Note that $f_{W_i}$ is the density function of the $m-\nu_i$th of the smallest order statistic from $m$ i.i.d. rv's $D_i,$ $i=1,\cdots,m,$ exponentially distributed with mean 1 $\left(W_i \buildrel d \over ={D_{m-\nu_i:m}} \right)$. For more details about order statistics, we refer readers to \citet{david2004order}. Furthermore, it is known that (see, e.g., \citet{basu19981}) 
\begin{eqnarray}
D_{m-\nu_i:m} \buildrel d \over =\sum\limits_{j=1}^{m-\nu_i}\frac{Q_{j}}{m-j+1},\label{23sf}
\end{eqnarray}
 where $Q_j$ are $m-\nu_i$ independent exponential distributions with mean 1. (Here, $\buildrel d \over =$ stands for the 'equality in distribution). 
 
 To derive the pdf of the aggregated random variable $S_n=X_1+\cdots+X_n$ with the mixed Bernstein copulas presented in Section \eqref{sectB}, the following result is needed.
 \begin{theorem}\label{ritazcmlfr} Assume that the rv's $Z_{1},\cdots,Z_{n}$ are dependent and defined with the Bernstein copulas and distributed as exponential $Exp(1)$, such as assumed in Section \eqref{sectB}. Then, the pdf of the aggregated random variable $\tilde{S}_n=Z_1+\cdots+Z_n$ is given by
\begin{eqnarray}
f_{\tilde{S}_n}(x)&=&\sum\limits_{l=n}^{\infty}A_l e^{-m x} x^{l-1}\frac{m^{l}}{\Gamma(l)}, \label{rezrionde}
\end{eqnarray}
where $A_l=\Pr\left({\sum\limits_{i=1}^{n}\sum\limits_{j=N_i+1}^{m}\Delta_{i,j}}=l\right)$ and $\Delta_{i,j}$ form a sequence of independent rv's, independent of $N_i,$ for $i=1,\cdots,n$ and $j=1,\cdots,m$
such that $\Delta_{i,j}$ follow a shifted geometric distribution $Geom(\frac{j}{m}),$  with probability mass function $\Pr\left(\Delta_{i,j}=l\right)=\frac{j}{m}\left(1-\frac{j}{m}\right)^{l-1},$ for for $i=1,\cdots,n,\,$ $j=1,\cdots,m$ and $l=1,2,\cdots.$
\end{theorem}
\begin{proof}From \eqref{23s} and \eqref{23sf}, the Laplace transform of $\tilde{S}_n$ is 
\begin{eqnarray*}
f_{\tilde{S}_n}^{\star}(s)&=&\displaystyle\sum\limits_{\nu_1=0}^{m-1}\cdots \displaystyle\sum\limits_{\nu_n=0}^{m-1}
\gamma\left(\frac{\nu_1}{m},\cdots,\frac{\nu_n}{m}\right)\prod\limits_{i=1}^{n}\prod\limits_{j=1}^{m-\nu_i}\frac{m-j+1}{m-j+1+s}\nonumber\\
&=&\displaystyle\sum\limits_{\nu_1=0}^{m-1}\cdots \displaystyle\sum\limits_{\nu_n=0}^{m-1}
\gamma\left(\frac{\nu_1}{m},\cdots,\frac{\nu_n}{m}\right)\prod\limits_{i=1}^{n}\prod\limits_{j=\nu_i+1}^{m}\frac{j}{j+s}.
\end{eqnarray*}
In fact, we can rewrite the above result as
\begin{eqnarray*}
f_{\tilde{S}_n}^{\star}(s)&=&\displaystyle\sum\limits_{\nu_1=0}^{m-1}\cdots \displaystyle\sum\limits_{\nu_n=0}^{m-1}
\gamma\left(\frac{\nu_1}{m},\cdots,\frac{\nu_n}{m}\right)\prod\limits_{i=1}^{n}\prod\limits_{j=1}^{m}\left(\frac{j}{j+s}\right)^{\mathbbm{1}_{ \left\{j\geq \nu_i+1 \right\}}},
\end{eqnarray*}
where $\mathbbm{1}_{ \left\{\mathcal{A} \right\}}$ denoting the indicator function on ${\mathcal{A}}$, we readily obtain that
\begin{eqnarray*}
f_{\tilde{S}_n}^{\star}(s)&=&\sum\limits_{\nu_1=0}^{m-1}\cdots\sum\limits_{\nu_n=0}^{m-1}\gamma\left(\frac{\nu_1}{m_1},\cdots,\frac{\nu_n}{m_n}\right)
\prod\limits_{j=1}^{m}\left(\frac{j}{j+s}\right)^{\sum\limits_{i=1}^{n}\mathbbm{1}_{ \left\{j\geq \nu_i+1 \right\}}}.
\end{eqnarray*}
On the other hand, we have $\prod\limits_{j=1}^{m}\left(\frac{j}{j+s}\right)^{\sum\limits_{i=1}^{n}\mathbbm{1}_{ \left\{j\geq \nu_i+1 \right\}}}=\left(\frac{m}{m+s}\right)^{mn-\sum\limits_{i=1}^{n}\nu_i}\prod\limits_{j=1}^{m}\left(\frac{\frac{j}{m}}{1-\left(1-\frac{j}{m}\right)\left(\frac{m}{m+s}\right)}\right)^{\sum\limits_{i=1}^{n}\mathbbm{1}_{ \left\{j\geq \nu_i+1 \right\}}}.$\\
$f_{\tilde{S}_n}^{\star}$ can then be rewritten as 
\begin{eqnarray}
f_{\tilde{S}_n}^{\star}(s)&=&E\left[\left(\frac{m}{m+s}\right)^{\sum\limits_{i=1}^{n}\sum\limits_{j=N_i+1}^{m}\Delta_{i,j}}\right]\nonumber\\
&=&\sum\limits_{l=0}^{\infty}\left(\frac{m}{m+s}\right)^{l}\Pr\left({\sum\limits_{i=1}^{n}\sum\limits_{j=N_i+1}^{m}\Delta_{i,j}}=l\right), \label{rezrosw}
\end{eqnarray}
where $\Delta_{i,j},$ form a sequence of independent rv's, independent of $N_i,$ for $i=1,\cdots,n$ and $j=1,\cdots,m$
such that $\Delta_{i,j}$ follow a shifted geometric distribution $Geom(\frac{j}{m}),$ for $i=1,\cdots,n$ and $j=1,\cdots,m$. The inversion of \eqref{rezrosw} with respect to $s$ yields \eqref{rezrionde} with 
\begin{eqnarray}
A_l&=&\sum\limits_{\nu_1=0}^{m-1}\cdots\sum\limits_{\nu_n=0}^{m-1}\gamma\left(\frac{\nu_1}{m_1},\cdots,\frac{\nu_n}{m_n}\right)\Pr\left(\sum\limits_{i=1}^{n}\sum\limits_{j=\nu_i+1}^{m}\Delta_{i,j}=l\right).\qquad\label{bemol}
\end{eqnarray}
Since the support of the rv ${\sum\limits_{i=1}^{n}\sum\limits_{j=N_i+1}^{m}\Delta_{i,j}}$ is $\{n,n+1,\cdots\},$ then $A_l=0$ for $l=0,\cdots,n-1.$ The expression in \eqref{rezrionde} follows immediately.
\end{proof} 
Now, we are in position to derive a closed-form expression of the pdf of $S_n=\sum\limits_{i=1}^{n}X_i.$ 

In the following corollary, we use Theorem \eqref{ritazcmlfr} to derive the pdf of the aggregated random variable $S_n=X_1+\cdots+X_n.$
\begin{corollary}\label{ritazcml} Let $\mathbf{X}=\left(X_{1},\cdots,X_{n}\right)$ be a vector of $n$ continuous and non-negative rv's defined by the stochastic representation \eqref{eqP2derfmlo}. Then, the pdf of the aggregated random variable $S_n$ is given by
\begin{eqnarray}
f_{S_n}(x)&=&\sum\limits_{l=n}^{\infty}A_l(-1)^l \frac{m^{l}}{\Gamma(l)}x^{l-1}{f_{\Theta}^{\star}}^{(l)}\left(m x\right), \label{rezro}
\end{eqnarray}
for $x>0.$ 
\end{corollary}
\begin{proof}From \eqref{eqP2derfmlo}, we have $f_{S_n}(x)=\int_{0}^{\infty}\theta f_{\tilde{S}_n}(x\theta)f_{\Theta}(\theta)d\theta.$ Substitution of \eqref{rezrionde} into the last expression yields the required result.
\end{proof} 
 \begin{remark} If $m=1$, It follows from \eqref{bemol} that 
\begin{eqnarray*}
A_l&=& \left\{
 \begin{array}{ll}
1& \mbox{if } l=n, \\
0 & \mbox{if } l\neq n.\label{rezrolkj12lk}
 \end{array}
\right.
\end{eqnarray*}
Substituting the last expression into \eqref{rezro}, one obtains the pdf of the aggregated risk $S_n$ discussed in \citet{sarabia2018aggregation} and given by
\begin{eqnarray*}
f_{S_n}(x)&=& \frac{(-1)^n}{\Gamma(n)}x^{n-1}{f_{\Theta}^{\star}}^{(n)}( x). 
\end{eqnarray*}
 \end{remark}
 We further derive the survival function of the distribution of $S_n=\sum\limits_{i=1}^{n}X_i.$
\begin{corollary}The survival function of the distribution of $S_n$ is given by
\begin{eqnarray*}
\Pr\left(S_n>x\right)&=&\sum\limits_{i=0}^{\infty}(-1)^{i} B_i \frac{(mx)^i}{i!}{f_{\Theta}^{\star}}^{(i)}\left( mx\right),
\end{eqnarray*}
for $x>0,$ where $B_i =\sum\limits_{\ell=\max(i+1,n)}^{\infty}A_{\ell}=\Pr\left({\sum\limits_{i=1}^{n}\sum\limits_{j=N_i+1}^{m}\Delta_{i,j}}\geq \max(i+1,n)\right).$
\end{corollary}
\begin{proof}From \eqref{rezro}, we have
\begin{eqnarray*}
\Pr\left(S_n>x\right)&=&\sum\limits_{l=n}^{\infty}A_l(-1)^l \frac{m^{l}}{\Gamma(l)}\int_{x}^{\infty} z^{l-1}{f_{\Theta}^{\star}}^{(l)}(m z)dz\\
&=&\sum\limits_{l=n}^{\infty}A_l \frac{(-1)^l}{\Gamma(l)}\int_{mx}^{\infty} z^{l-1}{f_{\Theta}^{\star}}^{(l)}(z)dz.
\end{eqnarray*}
By integration by parts (see \citet{hartman2014theory}), one gets
\begin{eqnarray*}
\Pr\left(S_n>x\right)&=&\sum\limits_{l=n}^{\infty}A_l \sum\limits_{i=0}^{l-1} (-1)^{i}\frac{(mx)^i}{i!}{f_{\Theta}^{\star}}^{(i)}( mx),
\end{eqnarray*}
which completes the proof.
\end{proof}
 \begin{remark} If $m=1$, then one gets the survival function of the aggregated risk $S_n$ discussed in \citet{sarabia2018aggregation} and given by
\begin{eqnarray*}
\Pr\left(S_n>x\right)&=& \sum\limits_{i=0}^{n-1} \frac{x^i}{i!} (-1)^{i} {f_{\Theta}^{\star}}^{(i)}(x). 
\end{eqnarray*}
 \end{remark}
\section{TVaR with the mixed Bernstein copulas\label{sectDD} }
In this section we derive an expression of the TVaR for $n$ dependent rv's $X_{1},\cdots,X_{n}$ joined by the mixed Bernstein copulas with generator ${f_{\Theta}^{\star}}.$
\begin{theorem}\label{rezromiofrsq}
Let $X_{1},\cdots,X_{n}$ $n$ dependent rv's joined by the mixed Bernstein copulas with generator 
${f_{\Theta}^{\star}}$, then the TVaR of the aggregate risk $S_n=\sum\limits_{i=1}^{n}X_i$ is 
\begin{eqnarray}
TVaR_\kappa(S_n)&=&\frac{1}{1-\kappa}\sum\limits_{\nu=1}^{\infty} { P_\nu}\frac{(-1)^{\nu+1} m^{\nu-1} VaR_\kappa^{\nu}(S_n)}{\nu!} {f_{\Theta}^{\star}}^{(\nu-1)}\left(m VaR_\kappa(S_n)\right)\nonumber\\
&+&\frac{n}{(1-\kappa)} \int_{m VaR_\kappa(S_n)}^{\infty} {f_{\Theta}^{\star}}(x)dx, \label{rezromiofd}
\end{eqnarray}
where $P_\nu= E\left[\left({\sum\limits_{i=1}^{n}\sum\limits_{j=N_i+1}^{m}\Delta_{i,j}}\right)\mathbbm{1}_{ \left\{\sum\limits_{i=1}^{n}\sum\limits_{j=N_i+1}^{m}\Delta_{i,j}\geq \max(\nu,n) \right\}}\right]$ and $\Delta_{i,j}$ form a sequence of independent rv's, independent of $N_i,$ for $i=1,\cdots,n$ and $j=1,\cdots,m$ such that $\Delta_{i,j}$ follow a shifted geometric distribution $Geom(\frac{j}{m})$.
\end{theorem}
\begin{proof}The combination of \eqref{fuaredsnb} and \eqref{rezro} yields
\begin{eqnarray}
TVaR_\kappa(S_n)&=&\frac{1}{1-\kappa}\sum\limits_{l=n}^{\infty}{A_l(-1)^l } \frac{m^{l}}{\Gamma(l)}\int_{VaR_\kappa(S_n)}^{\infty} x^{l}{f_{\Theta}^{\star}}^{(l)}(m x) (x) dx\nonumber\\
&=&\frac{1}{1-\kappa}\sum\limits_{l=n}^{\infty}\frac{A_l(-1)^l }{m\Gamma(l)} \int_{mVaR_\kappa(S_n)}^{\infty} x^{l}{f_{\Theta}^{\star}}^{(l)}(x)dx.\qquad \label{eqP1dghm}
\end{eqnarray}
On the other hand, by integration by parts (see \citet{spiegel2013mathematical}), one gets 
\begin{eqnarray*}
\int_{mVaR_\kappa(S_n)}^{\infty} x^{l} {f_{\Theta}^{\star}}^{(l)}(x)dx &=&\sum\limits_{i=1}^{l}(-1)^{l-i+1}\frac{m^i l!}{i!} VaR_\kappa^{i}(S_n){f_{\Theta}^{\star}}^{(i-1)}\left(m VaR_\kappa(S_n)\right)+(-1)^{l}l! \int_{m VaR_\kappa(S_n)}^{\infty} {f_{\Theta}^{\star}}(x)dx.
\end{eqnarray*}
Combining the last expression with \eqref{eqP1dghm} gives
\begin{eqnarray*}
TVaR_\kappa(S_n)&=&\frac{1}{1-\kappa}\sum\limits_{l=n}^{\infty} \sum\limits_{i=1}^{l}{l A_l}(-1)^{i+1}\frac{m^{i-1} }{i!} VaR_\kappa^{i}(S_n){f_{\Theta}^{\star}}^{(i-1)}\left(m VaR_\kappa(S_n)\right)\\
&+&\frac{1}{m(1-\kappa)}\sum\limits_{l=n}^{\infty}{lA_l } \int_{m VaR_\kappa(S_n)}^{\infty} {f_{\Theta}^{\star}}(x)dx. 
\end{eqnarray*}
Otherwise, from Theorem \eqref{ritazcmlfr}, we have $\sum\limits_{l=n}^{\infty}{lA_l }=E\left[\sum\limits_{i=1}^{n}\sum\limits_{j=N_i+1}^{m}\Delta_{i,j}\right]=\sum\limits_{i=1}^{n}\sum\limits_{j=1}^{m} \frac{m}{j} \Pr\left[N_i\leq j-1\right]=mn.$ Thus \eqref{rezromiofd} holds.
\end{proof}
\begin{corollary}If we take $m=1$, then the $TVaR_\kappa(S_n)$ in Theorem \eqref{rezromiofrsq} reduces to
\begin{eqnarray*}
TVaR_\kappa(S_n)&=&\frac{n}{1-\kappa} \sum\limits_{i=1}^{n}\frac{(-1)^{i+1} }{i!} VaR_\kappa^{i}(S_n){f_{\Theta}^{\star}}^{(i-1)}( VaR_\kappa(S_n))+\frac{n}{(1-\kappa)}\int_{ VaR_\kappa(S_n)}^{\infty} {f_{\Theta}^{\star}}(x)dx.
\end{eqnarray*}
\end{corollary}
\begin{proof}Substituting $m=1$ into \eqref{rezromiofrsq}, we get the desired result.
\end{proof} 
 To derive the TVaR-based contribution of risk $i$, $i=1,\cdots,n$ to the sum $S_n$ with the mixed Bernstein copulas presented in Section \eqref{sectB}, the following result is needed.
\begin{lemma}\label{bemola} Assume that the rv's $Z_{1},\cdots,Z_{n}$ are dependent and defined with the Bernstein copulas and distributed as exponential $Exp(1)$, such as assumed in Section \eqref{sectB}. Then, the pdf of the random vector $(Z_i,\tilde{S}_n-Z_i)$ is given by
\begin{eqnarray}
f_{Z_i,\tilde{S}_n-Z_i}(x,y)&=&\sum\limits_{k=1}^{\infty}\sum\limits_{l=n-1}^{\infty}q_{k,l}^{(i)} \tau_{k,m}(x) \tau_{l,m}(y), \label{rezrosw1}
\end{eqnarray}
for $x>0$ and $y>0,$ where $\tau_{k,\beta}\left( x\right) =\frac{\beta^k x
^{k-1}e^{-\beta x}}{\left( k-1\right) !},$ $\beta>0,$ 
\begin{eqnarray*}
q_{k,l}^{(i)}&=&\Pr\left(\sum\limits_{j=N_i+1}^{m}\Delta_{i,j}=k, {\sum\limits_{\substack{\nu=1 \\ \nu\neq i}}^{n}\sum\limits_{j=N_\nu+1}^{m}{\Delta}_{\nu,j}}=l\right),
\end{eqnarray*}
and $\Delta_{\nu,j}$ form a sequence of independent rv's, independent of $N_\nu,$ for $\nu=1,\cdots,n$ and $j=1,\cdots,m$
such that $\Delta_{\nu,j}$ follow a shifted geometric distribution $Geom(\frac{j}{m})$.
\end{lemma}
\begin{proof}Let $f_{Z_i,\tilde{S}_n-Z_i}^{\star}$ be the joint Laplace transform of the random vector $(Z_i,\tilde{S}_n-Z_i)$ defined by $f_{Z_i,\tilde{S}_n-Z_i}^{\star}(s_1,s_2)=\int_{0}^{\infty} \int_{0}^{\infty} e^{-(x s_1+ys_2)}f_{Z_i,\tilde{S}_n-Z_i}(x,y)dxdy.$ It follows from \eqref{23s} and \eqref{23sf}
\begin{eqnarray}
f_{Z_i,\tilde{S}_n-Z_i}^{\star}(s_1,s_2)&=&\displaystyle\sum\limits_{\nu_1=0}^{m-1}\cdots \displaystyle\sum\limits_{\nu_n=0}^{m-1}
\gamma\left(\frac{\nu_1}{m},\cdots,\frac{\nu_n}{m}\right)\left(\prod\limits_{j=1}^{m-\nu_i}\frac{m-j+1}{m-j+1+s_1}\right)\left(\prod\limits_{\substack{k=1 \\ k\neq i}}^{n}\prod\limits_{j=1}^{m-\nu_k}\frac{m-j+1}{m-j+1+s_2}\right)\nonumber\\
&=&\displaystyle\sum\limits_{\nu_1=0}^{m-1}\cdots \displaystyle\sum\limits_{\nu_n=0}^{m-1}
\gamma\left(\frac{\nu_1}{m},\cdots,\frac{\nu_n}{m}\right)\left(\prod\limits_{\ell=\nu_i+1}^{m}\frac{\ell}{\ell+s_1}\right)\left(\prod\limits_{\substack{k=1 \\ k\neq i}}^{n}\prod\limits_{j=\nu_k+1}^{m}\frac{j}{j+s_2}\right). \label{rezroswgtym2}
\end{eqnarray}
After some rearrangements, \eqref{rezroswgtym2} becomes
\begin{eqnarray}
f_{Z_i,\tilde{S}_n-Z_i}^{\star}(s_1,s_2)&=&\displaystyle\sum\limits_{\nu_1=0}^{m-1}\cdots \displaystyle\sum\limits_{\nu_n=0}^{m-1}
\gamma\left(\frac{\nu_1}{m},\cdots,\frac{\nu_n}{m}\right)\left(\prod\limits_{\ell=\nu_i+1}^{m}\frac{\ell}{\ell+s_1}\right)\left(\prod\limits_{j=1}^{m}\left(\frac{j}{j+s_2}\right)^{\sum\limits_{\substack{k=1 \\ k\neq i}}^{n}\mathbbm{1}_{ \left\{j\geq \nu_k+1 \right\}}}\right). \quad \qquad \label{rezroswgtym}
\end{eqnarray}
On the other hand, we have $\prod\limits_{\ell={\nu_i+1}}^{m}\left(\frac{\ell}{\ell+s_1}\right)=\left(\frac{m}{m+s_1}\right)^{m-\nu_i}\prod\limits_{\ell={\nu_i+1}}^{m}\left(\frac{\frac{\ell}{m}}{1-\left(1-\frac{\ell}{m}\right)\left(\frac{m}{m+s_1}\right)}\right)$ and $\prod\limits_{j=1}^{m}\left(\frac{j}{j+s_2}\right)^{{\sum\limits_{\substack{k=1 \\ k\neq i}}^{n}\mathbbm{1}_{ \left\{j\geq \nu_k+1 \right\}}}}=\left(\frac{m}{m+s_2}\right)^{
mn-m+\nu_i-\sum\limits_{j=1}^{n}\nu_j
}\prod\limits_{j=1}^{m}\left(\frac{\frac{j}{m}}{1-\left(1-\frac{j}{m}\right)\left(\frac{m}{m+s_2}\right)}\right)^{{\sum\limits_{\substack{k=1 \\ k\neq i}}^{n}\mathbbm{1}_{ \left\{j\geq \nu_k+1 \right\}}}}.$ Substituting the latter expressions in \eqref{rezroswgtym}, it follows that the joint Laplace transform $f_{Z_i,\tilde{S}_n-Z_i}^{\star}$ can be rewritten as
\begin{eqnarray}
&&f_{Z_i,\tilde{S}_n-Z_i}^{\star}(s_1,s_2)=E\left[\left(\frac{m}{m+s_1}\right)^{\sum\limits_{j=N_i+1}^{m}\Delta_{i,j}}\left(\frac{m}{m+s_2}\right)^{\sum\limits_{\substack{\nu=1 \\ \nu\neq i}}^{n}\sum\limits_{j=N_\nu+1}^{m}{\Delta}_{\nu,j}}\right],\quad \qquad \label{redezrosw}
\end{eqnarray}
where $\Delta_{\nu,j}$ form a sequence of independent rv's, independent of $N_\nu,$ for $\nu=1,\cdots,n$ and $j=1,\cdots,m$
such that $\Delta_{\nu,j}$ follow a shifted geometric distribution $Geom(\frac{j}{m})$.

Since the support of the rv ${\sum\limits_{\substack{\nu=1 \\ \nu\neq i}}^{n}\sum\limits_{j=N_\nu+1}^{m}{\Delta}_{\nu,j}}$ is $\{n-1,n,\cdots\},$ the inversion of \eqref{redezrosw} with respect to $s_1$ and $s_2$ yields the required result.
\end{proof}
Now, we are in a position to derive a closed-formula for the TVaR-based contribution of risk $i$, $i=1,\cdots,n$ to the sum $S_n$ with the mixed Bernstein copulas presented in Section \eqref{sectB}.
\begin{theorem}\label{redazfoypfrs}Let $X_{1},\cdots,X_{n}$ $n$ dependent rv's joined by the mixed Bernstein copulas with generator ${f_{\Theta}^{\star}}$. Then the TVaR-based contribution of risk $i$, $i=1,\cdots,n$ to the sum $S_n=X_1+\cdots+X_n$ 
at level $\kappa$, $0 <\kappa< 1$, is
\begin{eqnarray}
TVaR_\kappa(X_i;S_n)&=& \frac{1}{1-\kappa}\sum\limits_{\nu=1}^{\infty} {P}_{\nu}^{(i)}\frac{(- 1)^{\nu+1} m^{\nu-1} VaR_\kappa^\nu(S_n)
}{\nu !} {f_{\Theta}^{\star}}^{(\nu-1)}({m VaR_\kappa(S_n)})\nonumber \\
&+&\frac{1}{1-\kappa}\int_{m VaR_\kappa(S_n)}^{\infty} {f_{\Theta}^{\star}}(x)dx, \quad \quad\label{rezromiofrsfrqj}
\end{eqnarray}
where ${P}_{\nu}^{(i)}=E\left[\left({\sum\limits_{j=N_i+1}^{m}\Delta_{i,j}}\right)\mathbbm{1}_{ \left\{\sum\limits_{i=1}^{n}\sum\limits_{j=N_i+1}^{m}\Delta_{i,j}\geq \max(\nu,n) \right\}}\right]$ and $\Delta_{i,j}$ form a sequence of independent rv's, independent of $N_i,$ for $i=1,\cdots,n$ and $j=1,\cdots,m$
such that $\Delta_{i,j}$ follow a shifted geometric distribution $Geom(\frac{j}{m})$.
 \end{theorem}
\begin{proof}From \eqref{fuaredsnbr} and \eqref{eqP2derfmlo}, the capital attributed to the continuous distributed risk $i$ can be expressed as
\begin{eqnarray}
TVaR_\kappa(X_i;S_n)&=&\frac{1}{1-\kappa} E\left (X_i.1_{ \left\{S_n> VaR_\kappa(S_n) \right\}}\right)\nonumber\\
&=& \frac{1}{1-\kappa}\int_{0}^{\infty} \frac{1}{\theta}E\left (Z_i.1_{ \left\{\tilde{S}_n>\theta VaR_\kappa(S_n) \right\}}\right)f_{\Theta} (\theta)d\theta\nonumber\\
&=& \frac{1}{1-\kappa} \int_{0}^{\infty} \frac{1}{\theta}\int_{\theta VaR_\kappa(S_n)}^{\infty} \int_{0}^{s} x{f_{Z_i,\tilde{S}_n-Z_i}}(x,s-x)f_{\Theta} (\theta)dxds d\theta,\label{eqPfdfr}
\end{eqnarray}
where $\tilde{S}_n=\sum\limits_{j=1}^{n}Z_j$ and $f_{Z_i,\tilde{S}_n-Z_i}$ is the pdf of the random vector $(Z_i,\tilde{S}_n-Z_i)$. 

Using Lemma \eqref{bemola} leads to 
\begin{eqnarray}
 \int_{0}^{s} x f_{Z_{i},\tilde{S}_n-Z_i}(x,s-x)dx &=& \sum\limits_{k=1}^{\infty}\sum\limits_{l=n-1}^{\infty}q_{k,l}^{(i)}\frac{k}{m} \tau_{k+l+1,m}(s) = \sum\limits_{r=n}^{\infty}a_{r}^{(i)} \tau_{r+1,m}(s), \label{eqPfd}
\end{eqnarray}
where $ a_{r}^{(i)} = \sum\limits_{k=1}^{r-n+1}\frac{k}{m} q_{k,r-k}^{(i)}$. Consequently, inserting \eqref{eqPfd} into \eqref{eqPfdfr}, one gets
\begin{eqnarray*}
TVaR_\kappa(X_i;S_n)&=& \frac{1}{1-\kappa} \sum\limits_{r=n}^{\infty}a_{r}^{(i)}\int_{0}^{\infty} \frac{1}{\theta}\int_{\theta VaR_\kappa(S_n)}^{\infty} \tau_{r+1,m}(s) f_\Theta(\theta)ds d\theta\\
&=& \frac{1}{1-\kappa} \sum\limits_{r=n}^{\infty}a_{r}^{(i)} \sum\limits_{\nu=0}^{r}\int_{0}^{\infty} \frac{1}{\theta} \frac{\left(m \theta VaR_\kappa(S_n)\right)^\nu}{\nu !} e^{-m \theta VaR_\kappa(S_n)} f_\Theta(\theta) d\theta.
\end{eqnarray*}
Otherwise from \eqref{eqPfd}, $\sum\limits_{r=n}^{\infty}a_{r}^{(i)}= \int_{0}^{\infty} \int_{0}^{s} x f_{Z_{i},\tilde{S}_n-Z_i}(x,s-x)dx=1,$ it follows that
\begin{eqnarray*}
TVaR_\kappa(X_i;S_n)&=&- \frac{1}{1-\kappa} \sum\limits_{r=n}^{\infty}a_{r}^{(i)} \sum\limits_{\nu=1}^{r} \frac{\left(-m VaR_\kappa(S_n)\right)^\nu}{\nu !} {f_{\Theta}^{\star}}^{(\nu-1)}({m VaR_\kappa(S_n)})+\frac{1}{1-\kappa}\int_{m VaR_\kappa(S_n)}^{\infty} {f_{\Theta}^{\star}}(x)dx. \label{fuaredsnbder} 
\end{eqnarray*}
Inverting the order of summation w.r.t. the variables $r$ and $\nu$ in the last expression, one readily obtains
\begin{eqnarray*}
TVaR_\kappa(X_i;S_n)&=&- \frac{1}{1-\kappa}\sum\limits_{\nu=1}^{\infty} \sum\limits_{r=max(\nu,n)}^{\infty}a_{r}^{(i)} \frac{\left(- mVaR_\kappa(S_n)\right)^\nu}{\nu !} {f_{\Theta}^{\star}}^{(\nu-1)}({m VaR_\kappa(S_n)})\\
&+&\frac{1}{1-\kappa}\int_{m VaR_\kappa(S_n)}^{\infty} {f_{\Theta}^{\star}}(x)dx,
\end{eqnarray*}
where ${P}_{\nu}^{(i)}=\sum\limits_{r=max(\nu,n)}^{\infty}\sum\limits_{k=1}^{r-n+1}{k} q_{k,r-k}^{(i)}=E\left[\left({\sum\limits_{j=N_i+1}^{m}\Delta_{i,j}}\right)\mathbbm{1}_{ \left\{{\sum\limits_{\substack{\nu=1 }}^{n}\sum\limits_{j=N_\nu+1}^{m}{\Delta}_{\nu,j}}\geq \max(\nu,n) \right\}}\right],$ and then \eqref{rezromiofrsfrqj} immediately follows.
\end{proof}
\begin{corollary}If we take $m=1$, then \eqref{rezromiofrsfrqj} reduces to
\begin{eqnarray*}
TVaR_\kappa(X_i;S_n)&=& \frac{1}{1-\kappa}\sum\limits_{\nu=1}^{n} (-1)^{\nu+1}\frac{\left(VaR_\kappa(S_n)\right)^\nu}{\nu !} {f_{\Theta}^{\star}}^{(\nu-1)}({VaR_\kappa(S_n)})+\frac{1}{1-\kappa}\int_{VaR_\kappa(S_n)}^{\infty} {f_{\Theta}^{\star}}(x)dx.
\end{eqnarray*}
\end{corollary}
\begin{proof}If we take $m=1$, it follows from \eqref{rezromiofrsfrqj} that ${{{P}_{\nu}^{(i)}=1}}$ if $\nu=1,2,\cdots,n$ and $0$ if $\nu > n.$ Therefore, the proof is complete.
\end{proof}
\begin{corollary} Let $X_{1},\cdots,X_{n}$ $n$ dependent rv's joined by the mixed Bernstein copulas with generator ${f_{\Theta}^{\star}}$. Then 
\begin{eqnarray*}
TVaR_\kappa(S_n)&=&\sum\limits_{i=1}^{n}TVaR_\kappa(X_i;S_n).
\end{eqnarray*}
\end{corollary}
\begin{proof}The result follows easily from Theorems \eqref{rezromiofrsq} and \eqref{redazfoypfrs}.
 \end{proof} 
 \begin{corollary} Let $X_{1},\cdots,X_{n}$ $n$ dependent rv's joined by the mixed Bernstein copulas with generator ${f_{\Theta}^{\star}}$. If the joint	probability	mass function of the random vector $\left(N_1,\cdots,N_n\right)$ is exchangeable, then 
\begin{eqnarray*}
TVaR_\kappa(X_i;S_n)&=&\frac{1}{n}TVaR_\kappa(S_n).
\end{eqnarray*}
\end{corollary}
\begin{proof}The proof follows from Theorem \eqref{redazfoypfrs} by using elementary calculus.
 \end{proof} 
\section{Models\label{sectD}}
In this section, we present some results as consequences of our main results stated previously. We will consider dependent models with different claim distributions of the type Pareto and Gamma distributions, see \citet{Albrecher11} and \citet{sarabia2018aggregation} for more details.
\subsection{Pareto claims with Clayton copula dependence}
 We assume that $\Theta$ has a gamma distribution, $\Theta\sim Gamma(a,b ),$ with pdf $f_\Theta(\theta)=\frac{b ^a }{\Gamma(a )}\theta^{a -1}e^{-b \theta}$, and a Laplace transform $f_{\Theta}^{\star}$ defined by 
 \begin{eqnarray}
 f_{\Theta}^{\star} (x)=\left(1+\frac{x}{b}\right)^{-a },\quad a \geq 1. \label{rezriondelosqza}
 \end{eqnarray}
 It follow that $X_i \sim Pareto(a,b )$ with survival function given by $\bar{H}_{i}(x)=f_{\Theta}^{\star} (x)=\left(1+\frac{x}{b }\right)^{-a },\quad a \geq 1,$ for $i=1,\cdots,n.$ Using \eqref{rezriondelosqza}, the expression for $C\left( u_1,\cdots,u_n \right) $ in \eqref{ritazcmlop346h} turns into 
\begin{eqnarray}
C\left( u_1,\cdots,u_n \right)&=&\displaystyle\sum\limits_{\ell_1=0}^{m}\cdots\displaystyle\sum\limits_{\ell_n=0}^{m}b _{\ell_1,\cdots,\ell_n} \left(1+\sum\limits_{i=1}^{n} \ell_i (u_i^{-\frac{1}{a }}-1) \right)^{-a }, \label{ritazcmlop346} 
\end{eqnarray}
where
$b _{\ell_1,\cdots,\ell_n}=\displaystyle\sum\limits_{\nu_1=0}^{\ell_1}\cdots\displaystyle\sum\limits_{\nu_n=0}^{\ell_n}(-1)^{\sum\limits_{i=1}^{n}(\ell_i-\nu_i)}\left[\prod\limits_{i=1}^{n}\binom{m-\nu_i}{m-\ell_i} \right]\left[\prod\limits_{i=1}^{n}\binom{m}{\nu_i} \right]\alpha \left(\frac{\nu_1}{m},\cdots,\frac{\nu_n}{m}\right).$

Using \eqref{treyfre}, the expression for the survival function of $(X_{1},\cdots,X_{n})$ turns into
\begin{eqnarray}
\bar{H}\left(x_1,\cdots,x_n\right)&=&\displaystyle\sum\limits_{\ell_1=0}^{m}\cdots\displaystyle\sum\limits_{\ell_n=0}^{m} b _{\ell_1,\cdots,\ell_n} \left(1+\sum\limits_{i=1}^{n} \frac{\ell_i x_i}{b }\right)^{-a }, \label{ritazcmlop34622} 
\end{eqnarray}
which is the joint survival function of a Pareto distribution. 
 \begin{remark} Note that if $m=1$, then the survival function $\bar{H}\left(x_1,\cdots,x_n\right)$ in \eqref{ritazcmlop34622} reduces to the joint survival function of a Pareto type \rom{2} distribution proposed by \citet{arnold1983pareto,arnold2015pareto} and given by $\bar{H}\left(x_1,\cdots,x_n\right)= \left(1+\sum\limits_{i=1}^{n} \frac{ x_i}{b }\right)^{-a }.$
\end{remark}
In the following theorem, we give a close expression for the pdf of the aggregated risk for the special case of the Pareto claim with Clayton copula dependence.
\begin{theorem}\label{ritazc} Let $S_n=X_1+\cdots+X_n$ be the sum of $n$ dependent rv's with joint cdf defined by the multivariate mixed Bernstein copulas. Then the pdf of the aggregated random variable is given by
 \begin{eqnarray}
f_{S_n}(x)&=&\sum\limits_{\ell=n}^{\infty}A_l\frac{m^{l}x^{l-1}}{b ^l B(l,a ){\left(1+m x/b \right)}^{a +l}}, \label{rezrompogh}
\end{eqnarray}
for $x>0,$ where $B(l,a )=\frac{\Gamma(l)\Gamma(a )}{\Gamma(l+a )}$ denotes the Beta function.
\end{theorem}
\begin{proof}Taking the $l$th order derivative of \eqref{rezriondelosqza} with respect to $x$ yields
\begin{equation}
{f_{\Theta}^{\star}}^{(l)}(x)= \frac{(-1)^l\Gamma(a +l)}{b ^l\Gamma(a ){\left(1+x/b \right)}^{a +l}}, \quad l=0,1,\cdots. \label{redazDFni}
\end{equation}
 Substituting the last expression into \eqref{rezro}, we get the desired result.
\end{proof} 
 \begin{remark}\label{eqP2derfml} If we take $m=1$ in \eqref{rezrompogh}, then the pdf of the aggregated risk reduces to 
 \begin{eqnarray}
f_{S_n}(x)&=&\frac{x^{n-1}}{b ^n B(n,a ) {\left(1+x/b \right)}^{a +n}}, \label{redazDFnop}
\end{eqnarray}
which is the pdf of the aggregated risk $S_n$ with Pareto
marginal distribution with shape parameter $a $, scale parameter $b $ and Clayton survival
copula discussed in \citet{sarabia2016risk,sarabia2018aggregation}.
 \end{remark}
 \subsection{Gamma claims with dependence claims}
 Our next model is based on a Gamma claim distribution, $X_i \sim Gamma(a,\lambda)$, for $a \leq 1,$ it follows that the survival function of the claim $X_i$ is 
\begin{eqnarray}
\bar{H}_{i}(x)=f_{\Theta}^{\star} (x)=\frac{\Gamma(a,\lambda x)}{\Gamma(a )},\qquad a \leq 1.\label{rezrionkmsfr}
\end{eqnarray}
Using \eqref{ritazcmlop3469}, the multivariate survival function of $(X_{1},\cdots,X_{n})$ can be written as
\begin{eqnarray*}
\bar{H}\left(x_1,\cdots,x_n\right)&=&\displaystyle\sum\limits_{\ell_1=0}^{m}\cdots\displaystyle\sum\limits_{\ell_n=0}^{m} b _{\ell_1,\cdots,\ell_n} \frac{\Gamma\left(a,\lambda \sum\limits_{i=1}^{n} \ell_i x_i\right)}{\Gamma(a )}.
\end{eqnarray*}
We are now ready to apply the preceding result to derive a closed expression for the pdf of the aggregated risk with dependent 
Gamma claim.
\begin{theorem}\label{ritazc} Let $S_n=X_1+\cdots+X_n$ be the sum of $n$ dependent rv's with gamma marginal distributions and with joint cdf defined by the multivariate mixed Bernstein copulas. Then the pdf of the aggregated random variable can be written as a finite mixture of Gamma distributions 
\begin{eqnarray}
f_{S_n}(x)&=&\sum\limits_{k=1}^{\infty} \omega_{k} f_{\mathcal{G}(a +k-1,\lambda m)}(x), \label{eqPiuy}
\end{eqnarray}
where $\omega_{k}=\sum\limits_{l=max(k,n)}^{\infty} \frac{A_l \Gamma(a +k-1)}{\Gamma(k)\Gamma(l-k+1)\Gamma(a )}(-1)^{l-k}(a -1)_{l-k}.$
\end{theorem}
\begin{proof}According to \eqref{rezrionkmsfr}, the $l$th-order derivative of $f_{\Theta}^{\star}$ is given by
\begin{equation}
{f_{\Theta}^{\star}}^{(l)}(x)= -\frac{\lambda^a }{\Gamma(a )}\sum\limits_{k=0}^{l-1}\binom{l-1}{k}(-1)^{l-k-1}\lambda^{l-k-1}e^{-\lambda x}(a -1)_{k}x^{a -k-1},
\label{redazDFnipo}
\end{equation}
with $(a)_n=a(a-1)\cdots (a-n+1)$ is the Pochhammer symbol.
 Substituting \eqref{redazDFnipo} into \eqref{rezro}, one gets
\begin{eqnarray*}
f_{S_n}(x)&=&\sum\limits_{l=n}^{\infty} \frac{A_l}{\Gamma(l)\Gamma(a )}\sum\limits_{k=0}^{l-1}\binom{l-1}{k}(-1)^{k}m^{a +l-k-1}\lambda^{a +l-k-1}(a -1)_{k} x^{a +l-k-2} e^{-\lambda m x}\\
&=&\sum\limits_{l=n}^{\infty} \frac{A_l}{\Gamma(l)\Gamma(a )}\sum\limits_{k=1}^{l}\binom{l-1}{l-k}(-1)^{l-k}m^{a +k-1}\lambda^{a +k-1}(a -1)_{l-k} x^{a +k-2} e^{-\lambda m x}.
\end{eqnarray*}
Equation \eqref{eqPiuy} automatically follows.
\end{proof} 
 \begin{remark} Substituting $m=1$ into \eqref{eqPiuy}, one obtains,
 \begin{eqnarray}
f_{S_n}(x)&=&\sum\limits_{k=1}^{n} \omega_{k} f_{\mathcal{G}(a +k-1,\lambda )}(x), \label{redazDFnoplm}
\end{eqnarray}
where 
\begin{eqnarray*}
\omega_{k}&=& \left\{
 \begin{array}{ll}
\frac{ \Gamma(a +k-1)}{\Gamma(k)\Gamma(n-k+1)\Gamma(a )}(-1)^{n-k}(a -1)_{n-k}& \mbox{if } k=1,2,\cdots,n \\
0 & \mbox{if } k=n+1,n+2,\cdots,\label{rezrolkj12lksq}
 \end{array}
\right.
\end{eqnarray*}
then \eqref{redazDFnoplm} reduces to the pdf of the aggregated risk $S_n$ with gamma
marginal distributions discussed in \citet{sarabia2018aggregation}.
 \end{remark}
 \section{Numerical illustrations\label{sectEE}}
In this section, numerical examples are given to illustrate our findings. We assume that the rv $\Theta$ has a Gamma distribution, $Gamma(a,b )$. In the first example, the Bernstein copula is based on an exchangeable copula, while in the second one, we use a non exchangeable copula. 
\begin{example}
In this example, the values for the risk measures $VaR_{0.95}(S_{2})$ and $TVaR_{0.95}(S_{2})$ are computed for different values of $m$ in the following two cases 
\begin{description}
 \item[(i)] The copula $\alpha$ is comonotonic, i.e.,
 \begin{equation*}
\alpha(u_{1},u_{2})= min(u_{1},u_{2}).
\end{equation*}%
 \item[(ii)] The copula $\alpha$ is counter-comonotonic, i.e.,
 \begin{equation*}
\alpha(u_{1},u_{2})= max(u_{1}+u_{2}-1,0).
\end{equation*}%
\end{description}
The obtained results are displayed in Table \ref{tab1}. 
 \begin{table}[H]
\centering
\caption{Impact of $m$ on the VaR, TvaR in the case of an exchangeable copula}
\begin{tabularx}
{\textwidth}{bsssssss}
\Xhline{3.5\arrayrulewidth}
\Xhline{3.5\arrayrulewidth}
 Case (i)& $m=1$ & $m=5$ & $m=10$ &$m=20$ &$m=30$ &$m=40$ &$m=50$ \\
\Xhline{3.5\arrayrulewidth}
 $VaR_{0.95}(S_{2})$& 139.12 & 155.60& 159.76& 162.15 & 162.95 & 163.34& 163.55 \\
 $TVaR_{0.95}(S_{2})$ &205.30 & 233.06 & 241.33 & 247.00 & 249.30 & 250.57 & 251.37\\
\Xhline{3.5\arrayrulewidth}
\Xhline{3.5\arrayrulewidth}
Case (ii) & $m=1$ & $m=5$ & $m=10$ &$m=20$ &$m=30$ &$m=40$ &$m=50$ \\
\Xhline{3.5\arrayrulewidth}
 $VaR_{0.95}(S_{2})$& 139.12 & 123.41& 119.98 & 118.06 & 117.39 & 117.05 & 116.84\\
 $TVaR_{0.95}(S_{2})$ &205.30 & 178.71 & 173.63 & 170.91 & 169.98 & 169.51 & 169.22\\
\Xhline{3.5\arrayrulewidth}
\Xhline{3.5\arrayrulewidth}
\end{tabularx}
\label{tab1}
\end{table}
As expected, introducing a positive dependence (negative dependence) between the risks leads to a heavier (a lighter) tail for the aggregate risk $S_{2}$.
\end{example}
\begin{example}
In this example, two non-exchangeable copulas are considered
\begin{description}
 \item[(i)] A piece-wise copula based on two different Gaussian copulas 
 \begin{equation*}
\alpha(u_{1},u_{2})=\left \{
\begin{array}{ccc}
\tau C_{\phi}\left( \frac{u_{1}}{\tau },u_{2}; r_{1}\right) & if &
u_{1}\leqslant \tau \text{ } \\
\tau u_{2}+\left( 1-\tau \right) C_{\phi}\left( \frac{u-\tau }{1-\tau },u_{2}; r_{2}\right) & & Otherwise,%
\end{array}%
\right.
\end{equation*}%
where $C_{\phi}(.,.;r)$ is the Gaussian copula with parameter $r$. In our numerical computation, it is assumed that $\tau=0.5$, $r_{1}=-0.95$, and $r_{2}=0.95$.
 \item[(ii)] Following \citet{LIEBSCHER2008}, we consider a non-exchangeable copula based on transformations of two Clayton copulas (Cf. Equation (4) in (\citet{LIEBSCHER2008})
\begin{equation*}
 \alpha\left(u_{1}, u_{2}\right)=\left(1+\sum_{i=1}^{2}\left(u_{i}^{-\gamma \theta_{i}}-1\right)\right)^{-1 / \gamma}\left(1+\sum_{i=1}^{2}\left(u_{i}^{-\delta\left(1-\theta_{i}\right)}-1\right)\right)^{-1 / \delta},
\end{equation*}
with $\gamma=6$, $\delta=2$, $\theta_{1}=0.525$, and $\theta_{2}=0.3$. 
\end{description}
For these two copulas, we compute the $VaR_{0.95}(S_{2})$, $TVaR_{0.95}(S_{2})$, $TVaR_{0.95}(X_{1},S_{2})$, and $TVaR_{0.95}(X_{2},S_{2})$. The obtained values are given in Table \ref{tab2}. 
 \begin{table}[H]
\centering
\caption{Impact of $m$ on the VaR, TvaR in the case of non-exchangeable copulas}
\begin{tabularx}
{\textwidth}{bsssssss}
\Xhline{3.5\arrayrulewidth}
\Xhline{3.5\arrayrulewidth}
 Case (i) & $m=1$ & $m=5$ & $m=10$ &$m=20$ &$m=30$ &$m=40$ &$m=50$ \\
\Xhline{3.5\arrayrulewidth}
 $VaR_{0.95}(S_{2})$& 139.12 & 139.86 & 141.87 & 142.99& 143.31 & 143.43 & 143.49\\
 $TVaR_{0.95}(S_{2})$ &205.30 & 209.14 & 215.04 & 219.35 & 221.09 & 222.04 & 222.64 \\
 $TVaR_{0.95}(X_{1},S_{2})$ & 102.65 & 105.48 & 109.01 & 111.47 & 112.43 & 112.94 & 113.26\\
 $TVaR_{0.95}(X_{2},S_{2})$ &102.65 & 103.66 & 106.03 & 107.88 & 108.66 & 109.10& 109.38\\
\Xhline{3.5\arrayrulewidth}
\Xhline{3.5\arrayrulewidth}
Case (ii) & $m=1$ & $m=5$ & $m=10$ &$m=20$ &$m=30$ &$m=40$ &$m=50$ \\
\Xhline{3.5\arrayrulewidth}
$VaR_{0.95}(S_{2})$& 139.12 & 148.88 & 152.44 & 154.52 & 155.20 & 155.53 & 155.71\\
 $TVaR_{0.95}(S_{2})$ &205.30 & 222.08 & 229.17 & 234.16 & 236.19 & 237.30 & 238.01\\
 $TVaR_{0.95}(X_{1},S_{2})$ &102.65 & 110.99 & 114.51& 116.99 & 118.00& 118.56 & 118.91\\
 $TVaR_{0.95}(X_{2},S_{2})$ & 102.65 & 111.09 & 114.66& 117.16 & 118.18& 118.74& 119.10\\
\Xhline{3.5\arrayrulewidth}
\Xhline{3.5\arrayrulewidth}
\end{tabularx}
\label{tab2}
\end{table}
From these results, one can see that introducing the second layer of dependence (i.e., Bernstein copula) impacts the obtained values for the risk measures and the capital allocations. The fact that the dependence is non-exchangeable does not translate to a significant difference between the capital allocations $TVaR_{0.95}(X_{1},S_{2})$ and $TVaR_{0.95}(X_{2},S_{2})$ and this is due to the fact that $X_{1}$ and $X_{2}$ are id. We obtained similar results under different choices of $\alpha$ and marginal distributions. This pushes us to assume that the dependence structure only affects the level of the risk measures while the capital allocation (as a percentage) is mainly depending on the marginal risks.
\end{example}
\bibliography{revised-ref}

\begin{thebibliography}{32}
\expandafter\ifx\csname natexlab\endcsname\relax\def\natexlab#1{#1}\fi

\bibitem[{Albrecher et~al.(2011)Albrecher, Constantinescu and
  Loisel}]{Albrecher11}
\textsc{Albrecher, H.}, \textsc{Constantinescu, C.} and \textsc{Loisel, S.}
  (2011)  Explicit ruin formulas for models with dependence among risks.
\newblock \textit{Insurance: Mathematics and Economics,} \textbf{{\bf{48}}},
  265--270.

\bibitem[{Arnold(1983)}]{arnold1983pareto}
\textsc{Arnold, B.C.} (1983)  \textit{{P}areto Distributions}.
\newblock Fairland: International Cooperative Publishing House.

\bibitem[{Arnold(2015)}]{arnold2015pareto}
\textsc{Arnold, B.C.} (2015)  \textit{{P}areto Distributions}.
\newblock Chapman \& Hall/CRC Monographs on Statistics \& Applied Probability.

\bibitem[{Basu and Singh(1998)}]{basu19981}
\textsc{Basu, A.P.} and \textsc{Singh, B.} (1998)  Order statistics in
  exponential distribution.
\newblock \textit{Handbook of Statistics} \textbf{17}, 3--23.

\bibitem[{Cai and Li(2005)}]{cai2005conditional}
\textsc{Cai, J.} and \textsc{Li, H.} (2005)  Conditional tail expectations for
  multivariate phase-type distributions.
\newblock \textit{Journal of Applied Probability} \textbf{42}, 810--825.

\bibitem[{Chiragiev and Landsman(2007)}]{chirag2007}
\textsc{Chiragiev, A.} and \textsc{Landsman, Z.} (2007)  Multivariate pareto
  portfolios: Tce-based capital allocation and divided differences.
\newblock \textit{Scandinavian Actuarial Journal} \textbf{2007}, 261--280.

\bibitem[{Cossette et~al.(2013)Cossette, C{\^o}t{\'e}, Marceau and
  Moutanabbir}]{cossette2013multivariate}
\textsc{Cossette, H.}, \textsc{C{\^o}t{\'e}, M.P.}, \textsc{Marceau, E.} and
  \textsc{Moutanabbir, K.} (2013)  Multivariate distribution defined with
  farlie--gumbel--morgenstern copula and mixed erlang marginals: Aggregation
  and capital allocation.
\newblock \textit{Insurance: Mathematics and Economics} \textbf{52}, 560--572.

\bibitem[{C{\^o}t{\'e} and Genest(2019)}]{cote2019dependence}
\textsc{C{\^o}t{\'e}, M.P.} and \textsc{Genest, C.} (2019)  Dependence in a
  background risk model.
\newblock \textit{Journal of Multivariate Analysis} \textbf{172}, 28--46.

\bibitem[{Cottin and Pfeifer(2014)}]{cottin2014bernstein}
\textsc{Cottin, C.} and \textsc{Pfeifer, D.} (2014)  From bernstein polynomials
  to bernstein copulas.
\newblock \textit{J. Appl. Funct. Anal} \textbf{9}, 277--288.

\bibitem[{David and Nagaraja(2004)}]{david2004order}
\textsc{David, H.} and \textsc{Nagaraja, H.} (2004)  \textit{Order Statistics}.
\newblock Wiley Series in Probability and Statistics. Wiley.

\bibitem[{Dhaene et~al.(2008)Dhaene, Laeven, Vanduffel, Darkiewicz and
  Goovaerts}]{dhaene2008can}
\textsc{Dhaene, J.}, \textsc{Laeven, R.J.}, \textsc{Vanduffel, S.},
  \textsc{Darkiewicz, G.} and \textsc{Goovaerts, M.J.} (2008)  Can a coherent
  risk measure be too subadditive?
\newblock \textit{Journal of Risk and Insurance} \textbf{75}, 365--386.

\bibitem[{Diers et~al.(2012)Diers, Eling and Marek}]{diers2012dependence}
\textsc{Diers, D.}, \textsc{Eling, M.} and \textsc{Marek, S.D.} (2012)
  Dependence modeling in non-life insurance using the bernstein copula.
\newblock \textit{Insurance: Mathematics and Economics} \textbf{50}, 430--436.

\bibitem[{Furman and Landsman(2005)}]{furman2005risk}
\textsc{Furman, E.} and \textsc{Landsman, Z.} (2005)  Risk capital
  decomposition for a multivariate dependent gamma portfolio.
\newblock \textit{Insurance: Mathematics and Economics} \textbf{37}, 635--649.

\bibitem[{Furman and Landsman(2006)}]{furman_landsman_2006}
\textsc{Furman, E.} and \textsc{Landsman, Z.} (2006)  Tail variance premium
  with applications for elliptical portfolio of risks.
\newblock \textit{ASTIN Bulletin} \textbf{36}, 433--462.

\bibitem[{Furman and Landsman(2010)}]{furman2010multivariate}
\textsc{Furman, E.} and \textsc{Landsman, Z.} (2010)  Multivariate tweedie
  distributions and some related capital-at-risk analyses.
\newblock \textit{Insurance: Mathematics and Economics} \textbf{46}, 351--361.

\bibitem[{Furman and Zitikis(2008a)}]{furman2008weighted}
\textsc{Furman, E.} and \textsc{Zitikis, R.} (2008a)  Weighted risk capital
  allocations.
\newblock \textit{Insurance: Mathematics and Economics} \textbf{43}, 263--269.

\bibitem[{Hartman and Mikusinski(2014)}]{hartman2014theory}
\textsc{Hartman, S.} and \textsc{Mikusinski, J.} (2014)  \textit{The theory of
  Lebesgue measure and integration}  vol.~15.
\newblock Elsevier.

\bibitem[{Hurd et~al.(2007)Hurd, Salmon and Schleicher}]{hurd2007using}
\textsc{Hurd, M.}, \textsc{Salmon, M.} and \textsc{Schleicher, C.} (2007)
  Using copulas to construct bivariate foreign exchange distributions with an
  application to the sterling exchange rate index.
\newblock \textit{Bank of England Working Paper} .

\bibitem[{Kim(2007)}]{kim2007estimation}
\textsc{Kim, H.T.} (2007)  Estimation and allocation of insurance risk capital.
\newblock \textit{PhD thesis, University of Waterloo} .

\bibitem[{Kulpa(1999)}]{kulpa1999approximation}
\textsc{Kulpa, T.} (1999)  On approximation of copulas.
\newblock \textit{International Journal of Mathematics and Mathematical
  Sciences} \textbf{22}, 259--269.

\bibitem[{Landsman and Valdez(2003)}]{landsman2003tail}
\textsc{Landsman, Z.M.} and \textsc{Valdez, E.A.} (2003)  Tail conditional
  expectations for elliptical distributions.
\newblock \textit{North American Actuarial Journal} \textbf{7}, 55--71.

\bibitem[{Liebscher(2008)}]{LIEBSCHER2008}
\textsc{Liebscher, E.} (2008)  Construction of asymmetric multivariate copulas.
\newblock \textit{Journal of Multivariate Analysis} \textbf{99}, 2234--2250.

\bibitem[{Marri et~al.(2018)Marri, Ad{\'e}kambi and
  Moutanabbir}]{marri2018moments}
\textsc{Marri, F.}, \textsc{Ad{\'e}kambi, F.} and \textsc{Moutanabbir, K.}
  (2018)  Moments of compound renewal sums with dependent risks using mixing
  exponential models.
\newblock \textit{Risks} \textbf{6}, 86.

\bibitem[{Oakes(1989)}]{oakes1989bivariate}
\textsc{Oakes, D.} (1989)  Bivariate survival models induced by frailties.
\newblock \textit{Journal of the American Statistical Association,}
  \textbf{{\bf{84}}}, 487--493.

\bibitem[{Salmon et~al.(2006)Salmon, Schleicher et~al.}]{salmon2006pricing}
\textsc{Salmon, M.}, \textsc{Schleicher, C.} et~al. (2006)  Pricing
  multivariate currency options with copulas.
\newblock \textit{Copulas: From Theory to Application in Finance, Risk Books,
  London} .

\bibitem[{Sancetta and Satchell(2004)}]{sancetta2004bernstein}
\textsc{Sancetta, A.} and \textsc{Satchell, S.} (2004)  The bernstein copula
  and its applications to modeling and approximations of multivariate
  distributions.
\newblock \textit{Econometric theory} \textbf{20}, 535--562.

\bibitem[{Sarabia et~al.(2016)Sarabia, G{\'o}mez-D{\'e}niz, Prieto and
  Jord{\'a}}]{sarabia2016risk}
\textsc{Sarabia, J.M.}, \textsc{G{\'o}mez-D{\'e}niz, E.}, \textsc{Prieto, F.}
  and \textsc{Jord{\'a}, V.} (2016)  Risk aggregation in multivariate dependent
  pareto distributions.
\newblock \textit{Insurance: Mathematics and Economics} \textbf{71}, 154--163.

\bibitem[{Sarabia et~al.(2018)Sarabia, G{\'o}mez-D{\'e}niz, Prieto and
  Jord{\'a}}]{sarabia2018aggregation}
\textsc{Sarabia, J.M.}, \textsc{G{\'o}mez-D{\'e}niz, E.}, \textsc{Prieto, F.}
  and \textsc{Jord{\'a}, V.} (2018)  Aggregation of dependent risks in mixtures
  of exponential distributions and extensions.
\newblock \textit{ASTIN Bulletin} \textbf{48}, 1079--1107.

\bibitem[{Sklar(1959)}]{Sklar}
\textsc{Sklar, A.} (1959)  Fonctions de r\'{e}partition \`{a} n dimensions et
  leurs marges.
\newblock \textit{Publications de l'{I}nstitut de Statistique de
  l'{U}niversit\'{e} de {P}aris,} \textbf{{\bf{8}}}, 229--231.

\bibitem[{Spiegel et~al.(2013)Spiegel, Lipschutz and
  Liu}]{spiegel2013mathematical}
\textsc{Spiegel, M.R.}, \textsc{Lipschutz, S.} and \textsc{Liu, J.} (2013)
  \textit{Mathematical handbook of formulas and tables}.
\newblock McGraw-Hill.

\bibitem[{Tavin(2013)}]{tavin2013application}
\textsc{Tavin, B.} (2013)  \textit{Application of Bernstein copulas to the
  pricing of multi-asset derivatives}.
\newblock Springer.

\bibitem[{Yang et~al.(2020)Yang, Wang and Xie}]{yang2020bernstein}
\textsc{Yang, J.}, \textsc{Wang, F.} and \textsc{Xie, Z.} (2020)
  \textit{Bernstein Copulas and Composite Bernstein Copulas}.
\newblock Springer.

\end{thebibliography}

\end{document}